\documentclass[11pt]{amsart}
\usepackage[utf8]{inputenc}
\usepackage[T1]{fontenc}
\usepackage{lmodern}
\usepackage{longtable}
\usepackage{multicol}
\usepackage{bm,mathrsfs}
\usepackage{fancyvrb}
\usepackage[colorlinks=true,citecolor=blue,hypertexnames=false]{hyperref}
\usepackage{color}
\usepackage{amsmath,amssymb}
\usepackage{graphics}
\usepackage{url}

\usepackage{flushend} 

\usepackage{mdwmath}
\usepackage{mdwtab}
\usepackage{eqparbox}

\newcommand{\x}{x}

\newcommand{\Dx}{\partial}
\newcommand{\Wx}{\K[\x, \Dx]}
\newcommand{\Kx}{\K[\x]}

\newcommand{\bigO}{{\mathcal{O}}}
\newcommand{\bigOsoft}{\tilde{\mathcal{O}}}

\newcommand{\K}{\mathbb{K}}
\newcommand{\bN}{\mathbb{N}}

\newcommand{\V}{\mathbb{V}}

\newcommand{\mathe}{\mathrm{e}}

\def\gathen#1{{#1}}
\def\hoeven#1{{#1}}

\newtheorem{lem}{Lemma}
\newtheorem{thm}{Theorem}

\title[Quasi-optimal multiplication of linear differential operators]{Quasi-optimal multiplication of \\ linear differential operators}

\author[Alexandre Benoit]{Alexandre Benoit
}
\address{UPMC (France)}
\email{Alexandre.Benoit@lip6.fr}
\author[Alin Bostan]{Alin Bostan$\mbox{}^\dagger$}\thanks{$\dagger$ Supported by the Microsoft Research\,--\,INRIA Joint Centre.}
\address{INRIA (France)}
\email{Alin.Bostan@inria.fr}
\author[Joris van der Hoeven]{Joris van der Hoeven$\mbox{}^\ddagger$}\thanks{$\ddagger$ Supported by the ANR-09-JCJC-0098-01 {\textsc{MaGiX}} project, as well as by a
Digiteo 2009-36HD grant and R\'egion \^Ile-de-France.
}
\address{CNRS \& \'Ecole polytechnique (France)}
\email{vdhoeven@lix.polytechnique.fr}

\date{\today}
\begin{document}
	
\begin{abstract} We show that linear differential operators with polynomial
coefficients over a field of characteristic zero can be multiplied in
quasi-optimal time. This answers an open question raised by van der Hoeven. 
\end{abstract}

\maketitle

\section{Introduction} The product of \emph{polynomials\/} and the product of
\emph{matrices\/} are two of the most basic operations in mathematics; the
study of their computational complexity is central in computer science.
In this paper, we will be interested in the computational complexity of
multiplying two \emph{linear differential operators}. These algebraic objects
encode linear differential equations, and form a non-commutative ring that
shares many properties with the commutative ring of usual
polynomials~\cite{Ore32,Ore33}. The structural analogy between polynomials and
linear differential equations was discovered long ago by Libri and
Brassinne~\cite{Libri1833, Brassinne1864, Demidov83}. Yet, the algorithmic
study of linear differential operators is currently much less advanced than in
the polynomial case: the complexity of multiplication has been addressed only
recently~\cite{vdHoeven02,BoChLe08}, but not completely solved. The aim of the
present work is to make a step towards filling this gap, and to solve an open
question raised in~\cite{vdHoeven02}.

Let $\K$ be an effective field.
That is, we assume data structures for representing
the elements of $\K$ and algorithms for performing the field
operations. The aim of \emph{algebraic complexity theory\/} is to study the
cost of basic or more complex algebraic operations over $\K$ (such as the cost
of computing the greatest common divisor of two polynomials of degrees less
than $d$ in $\K [ x ]$, or the cost of Gaussian elimination on an
$r \times r$ matrix in $\K^{r \times r}$) in terms of the number of operations
in $\K$. The \emph{algebraic complexity\/} usually does not coincide with the
\emph{bit complexity}, which also takes into account the potential growth of
the actual coefficients in $\K$. Nevertheless, understanding the algebraic
complexity usually constitutes a first useful step towards understanding the
bit complexity. Of course, in the special, very important case when the
field~$\K$ is finite, both complexities coincide up to a constant factor.

The complexities of operations in the rings $\K[\x]$ and $\K^{r \times r}$
have been intensively studied during the last decades. It is well established
that polynomial multiplication is a \emph{commutative complexity yardstick},
while matrix multiplication is a \emph{non-commutative complexity yardstick},
in the sense that the complexity of operations in~$\K[\x]$ (resp.~in~$\K^{r
\times r}$) can generally be expressed in terms of the cost of multiplication
in~$\K[\x]$ (resp.~in~$\K^{r \times r}$), and for most of them, in a
quasi-linear way~\cite{AhHoUl74,BiPa94,BuClSh97,Pan00,GaGe03}.

Therefore, understanding the algebraic complexity of multiplication in
$\K[\x]$ and $\K^{r \times r}$ is a fundamental question. It is well known
that polynomials of degrees~$< d$ can be multiplied in time $\mathsf{M}
(d ) =\mathcal{O} ( d \log d \log \log d )$ using
algorithms based on the Fast Fourier Transform
(FFT)~\cite{CoTu65,ScSt71,CaKa91}, and two $r \times r$ matrices in~$\K^{r
\times r}$ can be multiplied in time $\mathcal{O} ( r^{\omega} )$,
with $2\leqslant \omega \leqslant 3$ {\cite{Strassen69,Pan84,CoWi90}}. The current tightest
upper bound, due to Vassilevska Williams~\cite{VassilevskaWilliams12}, is
$\omega < 2.3727$, following work of Coppersmith and Winograd~\cite{CoWi90}
and Stothers~\cite{Stothers10}. Finding the best upper bound on~$\omega$ is
one of the most important open problems in algebraic complexity theory.

In a similar vein, our thesis is that understanding the algebraic complexity
of multiplication of linear differential operators is a very important
question, since the complexity of more involved, higher-level operations on
linear differential operators can be reduced to that of
multiplication~\cite{vdHoeven11}.

{}From now on, we will assume that \emph{the base field~$\K$ has characteristic zero}.
Let $\Wx$ denote the associative algebra
$\K\langle\x,\Dx;\Dx\x=\x\Dx+1\rangle$ of linear differential operators in
$\Dx =\frac{d}{dx}$ with polynomial coefficients in $\x$. Any element $L$ of
$\Wx$ can be written as a finite sum $\sum_i L_i(x) \Dx^i$ for uniquely
determined polynomials $L_i$ in $\Kx$. We say that $L$ has bidegree less than
$(d,r)$ in $(\x,\Dx)$ if~$L$ has degree less than~$r$ in~$\Dx$, and if all
$L_i$'s have degrees~less than~$d$ in~$\x$. The degree in $\Dx$ of $L$ is
usually called the \emph{order\/} of~$L$.

The main difference with the commutative ring $\K[x,y]$ of usual bivariate
polynomials is the commutation rule $\Dx\x=\x\Dx+1$ that simply encodes, in
operator notation, Leibniz's differentiation rule $\frac{d}{dx}(\x
f)=\x\frac{d}{dx}(f)+f$. This slight difference between $\Wx$ and $\K[x,y]$
has a considerable impact on the complexity level. On the one hand, it is
classical that multiplication in $\K[x,y]$ can be reduced to that of
polynomials in $\K[x]$, due to a technique commonly called \emph{Kronecker's
trick}~\cite{Moenck76,GaGe03}. As a consequence, any two polynomials of
degrees less than $d$ in $x$, and less than~$r$ in $y$, can be multiplied in
\emph{quasi-optimal time\/} $\bigO(\mathsf{M}(dr))$. On the other hand,
under our hypothesis that $\K$ has characteristic zero,
it was shown by van der Hoeven~\cite{vdHoeven02} that
the product of two elements from $\Wx$ of bidegree
less than~$(n,n)$ can be computed in time $\bigO(n^\omega)$. Moreover, it has
been proved in~\cite{BoChLe08} that conversely, multiplication in $\K^{n
\times n}$ can be reduced to a constant number of multiplications in $\Wx$, in
bidegree less than~$(n,n)$. In other words, multiplying operators of
\emph{well-balanced bidegree\/} is computationally equivalent to 
matrix multiplication.

However, contrary to the commutative case, higher-level operations in $\Wx$,
such as the \emph{least common left multiple\/} (LCLM) and the \emph{greatest common
right divisor\/} (GCRD), do not preserve well-balanced
bidegrees~\cite{Grigoriev90,BoChLiSa12}. For instance, the LCLM of two
operators of bidegrees less than $(n,n)$ is of bidegree less than
$(2n(n+1),2n) = \bigO (n^2,n)$, and this bound is generically reached. This is
a typical phenomenon: operators obtained from computations in $\Wx$ tend to
have much larger degrees in~$x$ than in $\Dx$.

In the general case of operators with possibly unbalanced degrees~$d$~in $\x$
and~$r$ in $\Dx$, the naive algorithm has cost $\bigO(d^2 r^2 \min(d,r))$;
a better algorithm, commonly attributed to Takayama, has complexity          
$\bigOsoft(d r \min(d,r))$. We refer to~\cite[\S2]{BoChLe08} for a review of these algorithms.
When $r \leqslant d \leqslant r^{4-\omega}$, the best current upper bound for multiplication is $\bigO(r^{\omega-2}d^2)$~\cite{vdHoeven02,vdHoeven11}.
It was asked by van der Hoeven~\cite[\S6]{vdHoeven02}
whether this complexity could be lowered to $\bigOsoft(r^{\omega-1}d)$.
Here, and hereafter, the soft-O notation $\bigOsoft(\,)$
indicates that polylogarithmic factors in $d$ and in $r$ are neglected.
The purpose of the present work is to provide a positive answer to
this open question. Our main result is encapsulated in the following theorem:

\begin{thm}\label{th:main} 
	Let $\K$ be an effective field of characteristic zero.
	Operators in $\Wx$ of bidegree less than $(d,r)$ in
$(\x,\Dx)$ can be multiplied using \[ \bigOsoft (\min(d,r)^{\omega-2} \, dr )\]
operations in $\K$.\medskip
\end{thm}

Actually, we will prove slightly more refined versions of this theorem
(see Theorems~\ref{partial-1-th} and~\ref{partial-2-th} below),
by making the hidden log-terms in the complexity explicit.

In the important case $d \geqslant r$, our complexity bound reads $\bigOsoft
(r^{\omega-1} d)$. This is quasi-linear (thus quasi-optimal) with respect to
$d$. Moreover, by the equivalence result from~\cite[\S3]{BoChLe08}, the
exponent of $r$ is also the best possible. Besides, under the (plausible,
still conjectural) assumption that $\omega=2$, the complexity in
Theorem~\ref{th:main} is almost linear with respect to the output size. For
$r=1$ we retrieve the fact that multiplication in $\K[x]$ in degree $<d$ can
be done in quasi-linear time~$\bigOsoft(d)$; from this perspective, the result
of Theorem~\ref{th:main} can be seen as a generalization of the fast
multiplication for usual polynomials.

In an expanded version~\cite{SkewMul} of this paper, we will show
that analogues of Theorem~\ref{th:main} also hold for other types of
\emph{skew polynomials}. More precisely, we will prove similar complexity bounds
when the skew indeterminate $\partial: f(x) \mapsto f'(x)$ is replaced by the Euler
derivative $\delta: f(x) \mapsto x f'(x)$, or a shift operator $\sigma^c: f(x)
\mapsto f(x + c)$, or a dilatation $\chi_q: f(x) \mapsto f(qx)$.
Most of these other cases are treated by showing that rewritings such as
$\delta \leftrightarrow x \partial$ or $\sigma^c \leftrightarrow \exp (c \partial)$
can be performed efficiently. We will also prove complexity bounds
for a few other interesting operations on skew polynomials.

\medskip \noindent {\bf Main ideas.} The fastest known algorithms for
multiplication of usual polynomials in $\K[x]$ rely on an
evaluation-interpolation strategy at special points in the base
field~$\K$~\cite{CoTu65,ScSt71,CaKa91}.
This reduces polynomial multiplication  to the ``inner product'' in $\K$.
We adapt this strategy to the case of
linear differential operators in $\Wx$: the evaluation ``points'' are
\emph{exponential polynomials\/} of the form $x^n e^{\alpha x}$ on which
differential operators act nicely. With this choice, the evaluation and
interpolation of operators is encoded by \emph{Hermite evaluation and
interpolation\/} for usual polynomials (generalizing the classical Lagrange interpolation),
for which quasi-optimal algorithms exist.
For operators of bidegree less than $(d,r)$ in $(\x,\Dx)$, with
$r\geqslant d$, we use $p = \bigO ( r/d )$ evaluation points, and encode
the inner multiplication step by $p$ matrix multiplications in size~$d$. All
in all, this gives an FFT-type multiplication algorithm for differential
operators of complexity $\bigOsoft(d^{\omega-1}r)$. 
Finally, we reduce the case $r < d$ to the case $r \geqslant d$.
To do this efficiently, we design a fast algorithm for the computation of the
so-called \emph{reflection\/} of a differential operator,
a useful ring morphism that swaps the indeterminates $x$ and $\Dx$,
and whose effect is exchanging degrees and orders.

\section{Preliminaries}

Recall that $\K$ denotes an effective field of characteristic zero.
Throughout the paper, $\K[x]_{d}$ will denote the set
of polynomials of degree less than $d$ with coefficients in the field~$\K$,
and $\Wx_{d,r}$ will denote the set of linear differential operators in $\Wx$
with degree less than~$r$ in~$\Dx$, and polynomial coefficients in
$\K[x]_{d}$.

The cost of our algorithms will be measured by the number of field operations
in $\K$ they use. We recall that polynomials in $\K[x]_{d}$ can be multiplied
within $\mathsf{M}(d) = \bigO(d \, \log d\,\log\log d)= \bigOsoft (d)$
operations in~$\K$, using the FFT-based algorithms in~\cite{ScSt71,CaKa91},
and that $\omega$ denotes a feasible exponent for matrix multiplication
over~$\K$, that is, a real constant $2\leqslant \omega \leqslant 3$ such that
two $r\times r$ matrices with coefficients in $\K$ can be multiplied in time
$\bigO(r^\omega)$. Throughout this paper, we will make the classical
assumption that $\mathsf{M}(d) / d$ is an increasing function in $d$.

Most basic polynomial operations in $\K[x]_{d}$
(division, Taylor shift, extended gcd, multipoint evaluation, interpolation, etc.)
have cost $\bigOsoft(d)$~\cite{AhHoUl74,BiPa94,BuClSh97,Pan00,GaGe03}.
Our algorithms will make a crucial use of the following result due to
Chin~\cite{Chin76}, see also~\cite{OlSh00} for a formulation in terms of
structured matrices.

\begin{thm}[Fast Hermite evaluation--interpolation]\label{th:Hermite}
  Let~$\K$ be an effective field of characteristic zero, let $c_0,\ldots, c_{k-1}$ be $k$ positive integers, and let $d = \sum_i c_i$.
 Given~$k$ mutually distinct points $\alpha_0,\ldots, \alpha_{k-1}$ in~$\K$ and a polynomial $P\in \K[x]_{d}$, one can
  compute the vector of~$d$ values
  \begin{eqnarray*}
    \mathcal{H} &=&
     ( P(\alpha_0), P'(\alpha_0), \ldots, P^{(c_0-1)}(\alpha_0), \ldots\ldots,\\
    && \phantom{(} P(\alpha_{k-1}), P'(\alpha_{k-1}), \ldots, P^{(c_{k-1}-1)}(\alpha_{k-1}))
  \end{eqnarray*}
  in $\bigO(\mathsf{M}(d)\log k) = \bigOsoft(d)$ arithmetic operations in $\K$.
  Conversely, $P$ is uniquely determined by~$\mathcal{H}$, and its coefficients can be recovered from~$\mathcal{H}$ in
  $\bigO(\mathsf{M}(d)\log k) = \bigOsoft(d)$ arithmetic operations in $\K$.
\end{thm}

\section{The new algorithm in the case $r\geqslant d$} \label{sec:RgtD}

\subsection{Multiplication by evaluation and interpolation}\label{eval-interpol-sec}

Most fast algorithms for multiplying two polynomials $P, Q \in \K
[x]_d$ are based on the evaluation-interpolation strategy. The idea is to pick
$2 d - 1$ distinct points $\alpha_0, \ldots, \alpha_{2 d - 2}$
in~$\K$, and to perform the following three steps:
\begin{enumerate}
  \item (Evaluation) Evaluate $P$ and $Q$ at $\alpha_0, \ldots, \alpha_{2 d - 2}$.
  
  \item (Inner multiplication) Compute the values $(PQ) (\alpha_i) = P  (\alpha_i) Q (\alpha_i)$ for $i < 2 d - 1$.
  
  \item (Interpolation) Recover $PQ$ from  $(PQ) (\alpha_0), \ldots, (PQ) (\alpha_{2 d - 2})$.
\end{enumerate}

The inner multiplication step requires only $\mathcal{O} (d)$ operations.
Consequently, if both the evaluation and interpolation steps can be performed
fast, then we obtain a fast algorithm for multiplying $P$ and $Q$. For
instance, if $\K$ contains a $2^p$-th primitive root of unity
with $2^{p-1} \leqslant 2 d - 1 < 2^p$, then both evaluation
and interpolation can be performed in time $\mathcal{O} (d \log d)$ using the
Fast Fourier Transform~{\cite{CoTu65}}.

For a linear differential operator $L \in \K [ x, \partial
]_{d, r}$ it is natural to consider evaluations at powers of $x$
instead of roots of unity. It is also natural to represent the evaluation of
$L$ at a~suitable number of such powers by a matrix. More precisely, given $k
\in \bN$, we may regard $L$ as an operator from $\K [ x
]_k$ to $\K [ x ]_{k + d}$. We may also regard
elements of $\K [ x ]_k$ and $\K [ x
]_{k + d}$ as column vectors, written in the canonical bases with powers
of $x$. We will denote~by
\begin{eqnarray*}
  \Phi_L^{k + d, k} & = & \left(\begin{array}{ccc}
    L ( 1 )_0 & \cdots & L ( x^{k - 1} )_0\\
    \vdots &  & \vdots\\
    L ( 1 )_{k + d - 1} & \cdots & L ( x^{k - 1} )_{k +
    d - 1}
  \end{array}\right) 
\end{eqnarray*}
the matrix of the $\K$-linear map $L:\K [ x ]_k \rightarrow \K [ x ]_{k + d}$ with respect to these bases. Given two operators $K, L$ in
$\K [ x, \partial ]_{d, r}$, we clearly have
\begin{eqnarray*}
  \Phi_{KL}^{k + 2 d, k} & = & \Phi_K^{k + 2 d, k + d} \Phi_L^{k + d, k}, \quad \text{for all} \; k\geqslant 0 .
\end{eqnarray*}
For $k = 2 r$ (or larger), the operator $KL$ can be recovered from the matrix
$\Phi_{KL}^{2 r + 2 d, 2 r}$, whence the formula
\begin{eqnarray}
  \Phi_{KL}^{2 r + 2 d, 2 r} & = & \Phi_K^{2 r + 2 d, 2 r + d} \Phi_L^{2 r +
  d, 2 r}  \label{inner-mul}
\end{eqnarray}
yields a way to multiply $K$ and $L$. For the complexity analysis, we thus
have to consider the three steps:
\begin{enumerate}
  \item (Evaluation) Computation of $\Phi_K^{2 r + 2 d, 2 r + d}$ and of
  $\Phi_L^{2 r + d, 2 r}$ from $K$ and $L$.
  
  \item (Inner multiplication) Computation of the matrix product~(\ref{inner-mul}).
  
  \item (Interpolation) Recovery of $KL$ from $\Phi_{KL}^{2 r + 2 d, 2 r}$.
\end{enumerate}
In {\cite{vdHoeven02,BoChLe08}}, this multiplication method was applied with success
to the case when $d = r$. In this ``square case'', the following result was
proved in~{\cite[\S4.2]{BoChLe08}}.

\begin{lem}
  \label{square-lem}Let $L \in \K [ x, \partial ]_{d, d}$.
  Then
  \begin{enumerate}
    \item We may compute $\Phi_L^{2 d, d}$ as a function of $L$ in time
    $\mathcal{O} ( d \, \mathsf{M} ( d ) )$;
    
    \item We may recover $L$ from $\Phi_L^{2 d, d}$
    in time $\mathcal{O} ( d \, \mathsf{M} ( d ) )$.
  \end{enumerate}
\end{lem}

\subsection{Evaluation--interpolation at exponential polynomials} Assume
now that $r \geqslant d$. Then a straightforward application of the above
evaluation-interpolation strategy yields an algorithm of sub-optimal
complexity. Indeed, the matrix $\Phi_{KL}^{2 r + 2 d, 2 r}$ contains a lot of
redundant information and, since its mere total number of elements exceeds $r^2$, one cannot expect
a direct multiplication algorithm of quasi-optimal complexity
$\tilde{\mathcal{O}} (d^{\omega - 1} r)$.

In order to maintain quasi-optimal complexity in this case as well, the idea
is to evaluate at so called \emph{exponential polynomials\/} instead of ordinary
polynomials. More specifically, given $L \in \K [ x, \partial
]_{d, r}$ and $\alpha \in \K$, we will use the fact that~$L$
also operates nicely on the vector space $\K [ x ] \mathe^{\alpha
x}$. Moreover, for any $P \in \K [ x ]$, we have
\begin{eqnarray*}
  L ( P \mathe^{\alpha x} ) & = & L_{\ltimes \alpha} ( P
  ) \mathe^{\alpha x},
\end{eqnarray*}
where
\begin{eqnarray*}
  L_{\ltimes \alpha} & = & \sum_i L_i ( x )  ( \partial +
  \alpha )^i
\end{eqnarray*}
is the operator obtained by substituting $\partial + \alpha$ for $\partial$ in
$L = \sum_i L_i ( x ) \partial^i$.
Indeed, this is a consequence of the fact that, by Leibniz's rule:
\[ \Dx^i ( P \mathe^{\alpha x} ) = \left( \sum_{j \leqslant i} \binom{i}{j} 
\alpha^j \Dx^{i-j} P \right ) \mathe^{\alpha x} = ( \partial +
  \alpha )^i(P) \mathe^{\alpha x}. \]

Now let $p = \lceil r / d \rceil$ and let $\alpha_0, \ldots,
\alpha_{p - 1}$ be $p$ pairwise distinct points in~$\K$. For each
integer $k\geqslant 1$, we define the vector space
\begin{eqnarray*}
  \V_k & = & \K [ x ]_k \mathe^{\alpha_0 x}
  \oplus \cdots \oplus \K [ x ]_k \mathe^{\alpha_{p - 1}
  x}
\end{eqnarray*}
with canonical basis
\[ ( \mathe^{\alpha_0 x}, \ldots, x^{k - 1}
\mathe^{\alpha_0 x}, \mathord{\ldots \ldots}, \mathe^{\alpha_{p - 1} x},
\ldots, x^{k - 1} \mathe^{\alpha_{p - 1} x} ). \]
Then we may regard~$L$
as an operator from $\V_k$ into $\V_{k + d}$ and we will
denote by $\Phi_L^{[ k + d, k ]}$ the matrix of this operator with
respect to the canonical bases. By what precedes, this matrix is block
diagonal, with $p$ blocks of size $d$:
\begin{eqnarray*}
  \Phi_L^{[ k + d, k ]} & = & \left(\begin{array}{ccc}
    \Phi_{L_{\ltimes \alpha_0}}^{k + d, k} &  & \\
    & \ddots & \\
    &  & \Phi_{L_{\ltimes \alpha_{p - 1}}}^{k + d, k}
  \end{array}\right) .
\end{eqnarray*}
Let us now show that the operator $L$ is uniquely determined by the matrix
$\Phi_L^{[ 2 d, d ]}$, and that this gives rise to an efficient
algorithm for multiplying two operators in $\K [ x, \partial
]_{d, r}$.

\begin{lem}
  \label{partial-eval-interpol-lem}Let $L \in \K [ x, \partial
  ]_{d, r}$ with $r \geqslant d$. Then
  \begin{enumerate}
    \item We may compute $\Phi_L^{[ 2 d, d ]}$ as a function of $L$
    in time $\mathcal{O} ( d \, \mathsf{M} ( r ) \log r )$;
    
    \item We may recover $L$ from the matrix $\Phi_L^{[ 2 d, d ]}$ in time
    $\mathcal{O} ( d \, \mathsf{M} ( r ) \log r )$.
  \end{enumerate}
\end{lem}

\begin{proof}
  For any operator $L = \sum_{i < d, \, j < r} L_{i,j} x^i \partial^j $ in $\K [ x, \partial ]_{d, r}$, we
  define its truncation $L^{\ast}$ at order $\mathcal{O} ( \partial^d )$ by
  \begin{eqnarray*}
    L^{\ast} & = & \sum_{i, j < d} L_{i, j} x^i \partial^j .
  \end{eqnarray*}
  Since $L - L^{\ast}$ vanishes on $\K [ x ]_d$, we notice
  that $\Phi_L^{2 d, d} = \Phi_{L^{\ast}}^{2 d, d}$.
  
  If $L \in \K [ \partial ]_r$, then $L^{\ast}$ can be
  regarded as the power series expansion of $L$ at $\partial = 0$ and order
  $d$. More generally, for any $i \in \{ 0, \ldots, p - 1 \}$, the
  operator $L_{\ltimes \alpha_i}^{\ast} (\partial)  = L (\partial +
    \alpha_i)^{\ast}$ coincides with the Taylor series
  expansion at $\partial = \alpha_i$ and order $d$:
  \begin{eqnarray*}
    L_{\ltimes \alpha_i}^{\ast} (\partial) = L (\alpha_i)  +   L' (\alpha_i) \partial + \! \cdots \! + \! \tfrac{1}{(d - 1)
    !} L^{(d - 1)} (\alpha_i) \partial^{d - 1} .
  \end{eqnarray*}
  In other words, the computation of the truncated operators $L_{\ltimes
  \alpha_0}^{\ast}, \ldots, L_{\ltimes \alpha_{p - 1}}^{\ast}$ as a function
  of $L$ corresponds to a Hermite evaluation at the points $\alpha_i$, with
  multiplicity $c_i = d$ at each point $\alpha_i$.
  By~Theorem~\ref{th:Hermite}, this computation can be performed in time
  $\mathcal{O} ( \mathsf{M} ( pd ) \log  p ) =
  \mathcal{O} ( \mathsf{M} ( r ) \log r )$.
  Furthermore, Hermite interpolation allows us to recover $L$ from $L_{\ltimes
  \alpha_0}^{\ast}, \ldots, L_{\ltimes \alpha_{p - 1}}^{\ast}$ with the same
  time complexity $\mathcal{O} ( \mathsf{M} ( r ) \log r )$.
  
  Now let $L \in \K [ x, \partial ]_{d, r}$ and consider
  the expansion of $L$ in $x$
  \begin{eqnarray*}
    L ( x, \partial ) & = & L_0 ( \partial ) + \cdots +
    x^{d - 1} L_{d - 1} ( \partial ) .
  \end{eqnarray*}
  For each $i$, one Hermite evaluation of $L_i$ allows us to compute the
  $L_{\ltimes \alpha_j, i}^{\ast}$ with $j < p$ in time $\mathcal{O} (
  \mathsf{M} ( r ) \log r )$. The operators $L_{\ltimes
  \alpha_j}^{\ast}$ with $j < p$ can therefore be computed in time
  $\mathcal{O} ( d \, \mathsf{M} ( r ) \log r )$. By
  Lemma~\ref{square-lem}, we need $\mathcal{O} ( r \, \mathsf{M} ( d ) ) =
  \mathcal{O} ( d \, \mathsf{M} ( r ) )$ additional operations
  in order to obtain $\Phi_L^{[ 2 d, d ]}$. Similarly, given $\Phi_L^{[ 2 d, d ]}$,
  Lemma~\ref{square-lem} allows us to recover the operators $L_{\ltimes
  \alpha_j}^{\ast}$ with $j < p$ in time
  $\mathcal{O} ( d \, \mathsf{M} ( r )  )$.
  Using $d$ Hermite interpolations, we also
  recover the coefficients $L_i$ of $L$ in time
  $\mathcal{O} ( d \, \mathsf{M} ( r ) \log r )$.
\end{proof}

\begin{thm}
  \label{partial-1-th} Assume $r \geqslant d$ and let $K, L \in \K [ x, \partial ]_{d,
  r}$. Then the product $KL$ can be computed in time
  ${\mathcal{O} ( d^{\omega - 1} r + d \, \mathsf{M} ( r ) \log r
  )}$.
\end{thm}

\begin{proof}
  Considering $K$ and $L$ as operators in $\K [ x, \partial
  ]_{3 d, 3 r}$, Lemma~\ref{partial-eval-interpol-lem} implies that the
  computation of $\Phi_K^{[ 4 d, 3 d ]}$ and $\Phi_L^{[ 3 d, 2
  d ]}$ as a function of $K$ and $L$ can be done in time $\mathcal{O}
  ( d \, \mathsf{M} ( r ) \log r )$. The multiplication
  \begin{eqnarray*}
    \Phi_{KL}^{[ 4 d, 2 d ]} & = & \Phi_K^{[ 4 d, 3 d ]}
    \Phi_L^{[ 3 d, 2 d ]}
  \end{eqnarray*}
  can be done in time $\mathcal{O} ( d^{\omega} p ) =\mathcal{O}
  ( d^{\omega - 1} r )$. Lemma~\ref{partial-eval-interpol-lem}
  finally implies that we may recover $KL$ from $\Phi_{KL}^{[ 4 d, 2 d
  ]}$ in time $\mathcal{O} ( d \, \mathsf{M} ( r ) \log r
  )$.
\end{proof}

\section{The new algorithm in the case $d > r$} \label{sec:DgtR}

Any differential operator~$L\in\Wx_{d,r}$ can be written in a unique form \[ L
= \sum_{i<r,j<d} L_{i,j}x^j\partial^i, \quad \text{for some scalars\quad}
L_{i,j} \in \K. \] This representation, with $x$ on the left and $\Dx$ on the
right, is called the \emph{canonical form of\/}~$L$.

Let $\varphi:\Wx \rightarrow \Wx$ denote the map  defined by \[
\varphi \left(\sum_{i<r,j<d} L_{i,j}x^j\partial^i \right) = \sum_{i<r,j<d}
L_{i,j}\Dx^j (-x)^i.\] In other words, $\varphi$ is the unique $\K$-algebra
automorphism of $\Wx$ that keeps the elements of $\K$ fixed, and is defined on
the generators of $\Wx$ by $\varphi(x)=\Dx$ and $\varphi(\Dx)=-x$. We will
call $\varphi$ the \emph{reflection morphism\/} of $\Wx$. The map~$\varphi$
enjoys the nice property that it sends $\Wx_{d,r}$ onto $\Wx_{r,d}$. In
particular, to an operator whose degree is higher than its order, $\varphi$
associates a ``mirror operator'' whose order is higher than its degree.

\subsection{Main idea of the algorithm in the case $d > r$} If $d > r$, 
then the reflection morphism $\varphi$ is the key to our fast
multiplication algorithm for operators in $\Wx_{d,r}$, since it allows us to
reduce this case to the previous case when $r \geqslant d$. More precisely,
given $K, L$ in $\Wx_{d,r}$ with $d > r$, the main steps of the
algorithm are:

\begin{enumerate}
    \item[(S1)] compute the canonical forms of~$\varphi(K)$ and~$\varphi(L)$, 
    \item[(S2)] compute the product $M = \varphi(K) \varphi(L)$ of  operators~$\varphi(K) \in \Wx_{r,d}$ and~$\varphi(L)\in \Wx_{r,d}$, 
using the algorithm described in the previous section, and 
    \item[(S3)] return the (canonical form of the) operator ${KL = \varphi^{-1}(M)}$.
 \end{enumerate}

Since $d > r$, step (S2) can be performed in complexity $\bigOsoft(r^{\omega-1} d)$ using the results of Section~\ref{sec:RgtD}. In the next subsection, we will prove that both steps (S1) and (S3) can be performed in $\bigOsoft(rd)$ operations in $\K$. This will enable us to conclude the proof of Theorem~\ref{th:main}.

\subsection{Quasi-optimal computation of reflections} We now show that the
\emph{reflection\/} and the \emph{inverse reflection\/} of a differential
operator can be computed quasi-optimally. The idea is that performing
reflections can be interpreted in terms of Taylor shifts for polynomials,
which can be computed in quasi-linear time using the algorithm from~\cite{AhStUl75}.

A first observation is that the composition $\varphi \circ \varphi$ is equal
to the involution $\psi:\Wx \rightarrow \Wx$ defined by \[ \psi
\left(\sum_{i<r,j<d} L_{i,j}x^j\partial^i \right) = \sum_{i<r,j<d} (-1)^{i+j}
L_{i,j}x^j\partial^i.\] As a direct consequence of this fact, it follows that
the map $\varphi^{-1}$ is equal to $\varphi \circ \psi$. Since $\psi(L)$ is
already in canonical form, computing $\psi(L)$ only consists of sign changes,
which can be done in linear time $\bigO(dr)$. Therefore, computing the inverse
reflection $\varphi^{-1}(L)$ can be performed within the same cost as
computing the direct reflection $\varphi(L)$, up to a linear overhead
$\bigO(rd)$.

In the remainder of this section, we focus on the fast computation of direct
reflections. The key observation is encapsulated in the next lemma. Here, and
in what follows, we use the convention that the entries of a matrix
corresponding to indices beyond the matrix sizes are all zero.
\begin{lem}\label{lem:key}
	Assume that $(p_{i,j})$ and $(q_{i,j})$ are two matrices in $\K^{r \times d}$ such that
	\[ \sum_{i,j} q_{i,j} x^i \Dx^j = \sum_{i,j} p_{i,j} \Dx^j x^i.\]
Then
 \[ i! \, q_{i,j} = \sum_{k \geqslant 0} \binom{j+k}{k} (i+k)! \, p_{i+k,j+k},\] 
where we use the convention that $p_{i,j}=0$ as soon as $i\geqslant r$ or $j \geqslant d$.
\end{lem}

\begin{proof} Leibniz's differentiation rule implies the commutation rule
\begin{eqnarray*} \partial^j \frac{x^i}{i!} & = & \sum_{k = 0}^j \binom{j}{k}
\frac{x^{i - k}}{( i - k ) !} \partial^{j - k}. \end{eqnarray*}
Together with the hypothesis, this implies the equality 
\begin{eqnarray*}
	\sum_{i,j}
(i!\, q_{i,j}) \frac{x^i}{i!} \Dx^j \!\! & = & \!\!\! \sum_{i,j}
(i!\,p_{i,j}) \Dx^j \frac{x^i}{i!} \\ \!\! & = & \!\!\!  \sum_{k \geqslant 0}\left(
\sum_{i,j} (i!\,p_{i,j}) \binom{j}{k} \frac{x^{i - k}}{( i - k
) !} \partial^{j - k} \right). 
\end{eqnarray*}
We conclude by extraction of coefficients.
\end{proof}

\begin{thm} \label{th:refl} 
  Let $L \in \Wx_{d, r}$.
  Then we may compute $\varphi(L)$ and $\varphi^{-1}(L)$ using
  $\mathcal{O} ( \min ( d \, \mathsf{M} ( r ), r \, \mathsf{M} ( d ) ) )= \bigOsoft(rd)$ operations in $\K$.
\end{thm}

\begin{proof} 
   We first deal with the case $r \geqslant d$.~If ${L = \sum_{i<r,\, j<d} \, p_{i,j} \, x^j \, \Dx^i}$, then by the first equality of Lemma~\ref{lem:key},
the reflection $\varphi(L)$ is equal to

\[\varphi(L) = \sum_{i<r,j<d} p_{i,j}\Dx^j (-x)^i = \sum_{i<r,j<d} q_{i,j}
(-x)^j \Dx^i,\] where 
\begin{eqnarray} \label{eq:refl} 
	i! \, q_{i,j} = \sum_{\ell \geqslant 0} \binom{j+\ell}{j} (i+\ell)!
\, p_{i+\ell,j+\ell}.
\end{eqnarray}
  For any fixed $k$ with $1-r\leqslant k \leqslant d-1$, let us introduce $G_k = \sum_i i!q_{i, i + k} x^{i + k}$ and $F_k
  = \sum_i i! p_{i, i + k} x^{i + k}$. These polynomials belong to $\K[x]_{d}$, since $p_{i,j}=q_{i,j}=0$ for $j\geqslant d$.
If $k \leqslant 0$, then Equation~\eqref{eq:refl} translates into
  \begin{eqnarray*}
    G_k ( x ) & = & F_k ( x + 1 ).
  \end{eqnarray*}
Indeed, Equation~\eqref{eq:refl} with $j=i+k$ implies that $G_k(x)$ is equal to
 \begin{eqnarray*}
\sum_{i,\ell} \binom{i+k+\ell}{i+k}(i+\ell)! \, p_{i+\ell,i+k+\ell} \, x^{i+k} & \\
= \sum_{j,s} j! p_{j,j+k}\binom{j+k}{s}x^s = F_k(x+1).&  
\end{eqnarray*}

Similarly, if $k > 0$, then the coefficients of $x^i$ in $G_k
  ( x )$ and $F_k ( x + 1 )$ still coincide for all $i \geqslant k$. 
In particular, we may compute $G_{1 - r}, \ldots, G_{d - 1}$
  from $F_{1 - r}, \ldots, F_{d - 1}$ by means of~$d + r\leqslant 2r$
  Taylor shifts of polynomials in $\K[x]_{d}$. Using the fast algorithm for Taylor shift in~\cite{AhStUl75}, this can be done in time 
$\mathcal{O} ( r \, \mathsf{M} ( d ) )$.

Once the coefficients of the $G_k$'s are available, the computation of the coefficients of $\varphi(L)$ requires $\bigO(dr)$ additional operations.

If~$d > r$, then we notice that the equality~\eqref{eq:refl} is equivalent to
\[ j! \, q_{i,j} = \sum_{\ell \geqslant 0} \binom{i+\ell}{i} (j+\ell)!
\, p_{i+\ell,j+\ell},\]
as can be seen by expanding the binomial coefficients.
Redefining $G_k := \sum_i i! q_{i+k,i}x^{i+k}$ and $F_k := \sum_{i}i!p_{i+k,i}x^{i+k}$,
similar arguments as above show that $\varphi(P)$ can be computed
using $\bigO( d \, \mathsf{M} ( r ) )$ operations in~$\K$.

By what has been said at the beginning of this section,
we finally conclude that the inverse reflection ${\varphi^{-1}(L) = \varphi(\psi(L))}$
can be computed for the same cost as the direct reflection $\varphi(L)$.
\end{proof}

\subsection{Proof of Theorem~\ref{th:main} in the case $d > r$} 
We will prove a slightly better result:

\begin{thm}
  \label{partial-2-th} Assume $d > r$ and  $K, L \in \K [ x, \partial ]_{d, r}$.
  Then the product $KL$ can be computed using
  ${\mathcal{O} ( r^{\omega - 1} d + r \, \mathsf{M} ( d ) \log d  )}$ operations in $\K$.
\end{thm}

\begin{proof}
Assume that
$K$ and $L$ are two operators in $\Wx_{d,r}$ with $d > r$. Then $\varphi(K)$
and $\varphi(L)$ belong to $\Wx_{r,d}$, and their canonical forms can be
computed in $\bigO(d\mathsf{M} ( r ))$  operations by Theorem~\ref{th:refl}. Using the
algorithm from Section~\ref{sec:RgtD}, we may compute $M = \varphi(K)
\varphi(L)$ in
$\mathcal{O} ( r^{\omega - 1} d + r \,\mathsf{M} ( d ) \log d )$
operations. Finally,
$KL=\varphi^{-1}(M)$ can be computed in $\bigO(r\,\mathsf{M} ( d ))$ operations by
Theorem~\ref{th:refl}. We conclude by adding up the costs of these
three steps.
\end{proof}

\section*{Acknowledgment}

The authors would like to thank the three referees for their useful remarks.



\def\cprime{$'$}

\end{document}